\documentclass{article}
\usepackage{authblk}
\usepackage{vmargin}

	\usepackage{amsmath}
	\usepackage{amssymb}
	\usepackage{hyperref}
	\usepackage{stmaryrd}
	\usepackage{bussproofs}
	\usepackage{array}
	\usepackage{url}
	\usepackage{graphicx}
	\usepackage{listings}

	\usepackage{paralist}	
	\usepackage{bm}
	\usepackage{amsthm}

	\newcommand{\sleq}{\sqsubseteq}
	
	\newcommand{\sSup}{\bigsqcup}
	\newcommand{\N}{\mathbb{N}}

	\newcommand{\tracei}[2]{\mathrm{Tr}^{{#1}}(#2)}

	\newtheorem{mythm}{Theorem}
	\newtheorem{myprop}[mythm]{Proposition}
	\newtheorem{mylem}[mythm]{Lemma}
	
	\theoremstyle{definition}
	\newtheorem{mydef}{Definition}

	\newcommand{\g}[1]{{\mathsf{#1}}}
	\newcommand{\iter}{\mathrm{iter}}

        \newcommand{\fork}{\mathord{\curlywedge}}  
        
		\newcommand{\join}{\curlyvee}

\begin{document}



\title{Diagrammatic Semantics for Digital Circuits}
\author[1]{Dan R. Ghica}
\author[1]{Achim Jung}
\author[2]{Aliaume Lopez}
\affil[1]{University of Birmingham}
\affil[2]{ENS Cachan, Universit\'e Paris-Saclay}
\maketitle

\begin{abstract}

We introduce a general diagrammatic theory of digital circuits, based on connections between monoidal categories and graph rewriting. The main achievement of the paper is conceptual, filling a foundational gap in reasoning syntactically and symbolically about a large class of digital circuits (discrete values, discrete delays, feedback). This complements the dominant approach to circuit modelling, which relies on simulation. The main advantage of our symbolic approach is the enabling of automated reasoning about \textit{parametrised circuits}, with a potentially interesting new application to \textit{partial evaluation} of digital circuits. Relative to the recent interest and activity in categorical and diagrammatic methods, our work makes several new contributions. The most important is establishing that categories of digital circuits are Cartesian and admit, in the presence of feedback expressive iteration axioms. The second is producing a general yet simple graph-rewrite framework for reasoning about such categories in which the rewrite rules are computationally efficient, opening the way for practical applications. 

\end{abstract}

\section{Introduction}

Of the many differences between the worlds of software and hardware design, a particularly intriguing one is their prevailing modelling methodologies. The workhorse of software reasoning ---~\textit{operational semantics}~\cite{DBLP:journals/jlp/Plotkin04a}~--- is syntactic and reduction-based. It is essentially an abstract, entirely machine-independent presentation of a programming language which is not required to be faithful to the execution model other than insofar as the final result is concerned. On the other hand, reasoning about hardware relies on having an accurate execution model, akin to what we would call an \textit{abstract machine} in programming languages, usually some kind of automaton~\cite{DBLP:journals/tcad/KurshanM91}. To reason about a circuit, it is translated so that its execution is simulated by the automaton. The abstract machine approach is of course established and useful in programming language theory as well~\cite{landin1965abstract}, especially in compiler design. But the operational semantics has several advantages over the abstract machine approach, of which perhaps the most important is the ability to evaluate programs which are specified only in part. This is useful because many front-end compiler optimisations are, in one way or another, partial evaluations~\cite{consel1993tutorial}.

Broadly speaking, the main contribution of our paper is to provide an operational semantics for digital circuits, based on diagram rewriting. Our methodology is influenced by the interplay between graph rewriting and monoidal categories, which led in the last decade to diagrammatic models for quantum computing~\cite{DBLP:conf/lics/AbramskyC04}, signal flow~\cite{DBLP:conf/popl/BonchiSZ15} and asynchronous circuits~\cite{DBLP:conf/birthday/Ghica13}. Algebraic specifications in the style of monoidal categories have been pioneered by Sheeran in the 1980s~\cite{DBLP:conf/lfp/Sheeran84} and a certain amount of algebraic reasoning about circuits using such specifications has been attempted~\cite{Luk1993}. However, a full and systematic categorical presentation of digital circuits has only been given recently by the first two authors of the present paper~\cite{GhicaJ16}. Starting only from an equational specification of digital components (``gates'') it shows that the free traced monoidal category, subject to certain quotients, is Cartesian. Such categories, known as \textit{dataflow categories}~\cite[Sec.~6.4]{selinger2010survey} (or \textit{traced cartesian categories}~\cite{hasegawa2012models}), have very useful equations for iterative unfolding of the trace~\cite{cuazuanescu1990towards,hasegawa2012models,simpson2000complete}, offering a convenient way to model feedback.

The main theoretical contribution of this paper is providing a rewriting semantics for dataflow categories with a discrete delay operator. It is well known that an algebraic semantics does not automatically translate into an operational semantics, because distributive laws (in particular the functoriality of the tensor product) are directionless. This is where the diagrammatic approach can help when used as a graphical syntax, by avoiding the need for such problematic laws~\cite{DBLP:conf/lics/BonchiGKSZ16} and leading to computationally efficient rewriting. The iteration axioms also raise difficulties, this time of identifying diagrammatic redexes. This problem is compounded by the fact that choosing the wrong iterators to unfold can lead to unproductive rewrites. Finally, the presence of delays raises yet a different set of technical challenges because they cannot be rewritten out of a circuit but only moved around using a \textit{retiming} axiom~\cite{leiserson1991retiming}. We solve these problems by writing circuit diagrams in a particular canonical form, which we call \textit{global trace delay}, for which we can provide effective and efficient unfolding, with certain guarantees of productivity. 

The main motivation of this work is to open the door to new optimisation techniques for digital circuits, similar to partial evaluation. We will test our theory against a particularly challenging class of circuits, so called \textit{circuits with combinational feedback}~\cite{DBLP:conf/iccad/Malik93}. These are circuits which, despite the presence of feedback loops, behave just like combinational circuits, i.e. they exhibit none of the effects associated with genuine feedback, such as state or oscillation. As is the case with operational semantics, we will see how handling such circuits is mathematically elementary and fully automated. This is indeed remarkable, because the conventional automata-based reasoning method does not accept combinational feedback. Denotational semantics can model such circuits~\cite{DBLP:journals/fmsd/MendlerSB12} but using rather complex mathematical machinery. Moreover, we will show how circuits with combinational feedback which are parametrised by unspecified ``black box'' components can be just as easily handled by our approach. As far as we know, there is no existing method for modelling such circuits (called ``abstract circuits'') in the design literature.

\section{Categorical semantics}\label{sec:catsem}

The material in most of this section is presented more extensively in~\cite{GhicaJ16}.

\subsection{Combinational circuits}

We introduce a categorical language of circuits. Let \emph{object variables}, labelling (collections of) wires, be natural numbers and let \emph{morphism variables} be labels for boxes (e.g., gates and circuits). This is a category of PROducts and Permutations (PROP)~\cite{lack2004composing}.

\begin{mydef}\label{def:circ}\em
	Let $\mathbf{Circ}$ be a categorical signature with objects the natural numbers~$\mathbb{N}$ and a finite set of morphisms which may be grouped into the following three classes:
	\begin{compactitem}
		\item \emph{levels} (or \textit{values}) $v:0\rightarrow 1$ forming a lattice $(\mathbf V,\sleq)$;
		\item \emph{gates} $k:m\rightarrow 1$; and
		\item the special morphisms \emph{join} $\join:2\rightarrow 1$, \emph{fork} $\fork:1 \rightarrow 2$, and \emph{stub} $\g w:1\rightarrow 0$. 
	\end{compactitem}
\end{mydef}
All circuit signatures include combinators for joining two outputs (\emph{join}) and duplicating an input (\emph{fork}), as well as the ability to discard an output (\emph{stub}). What varies from signature to signature is the number of signal levels and the set of gates. Since levels form a lattice, they must include a smallest element ($\bot$), corresponding to a disconnected input, and a top element ($\top$) corresponding to an illegal output (``short circuit''). In the simplest and most common instance, the set of level has two other elements, \emph{high} and \emph{low}, but it can go beyond that. For example, in the case of metal-oxide-semiconductor field-effect transistors (MOSFET) it makes sense, in certain designs, to model the diode properties of the transistor by taking into account four levels (strong and weak high and low voltage, cf. the relevant IEEE standard for logical simulations~\cite{IEEE1164}).

\emph{Circuits}, in the sense of this paper, are the morphisms of a free categorical construction over their signature. Beginning with  \emph{combinational circuits}, the free construction is as follows:
\begin{mydef}\label{def:ccirc}
	Let $\mathbf{CCirc}$ be the free symmetric monoidal category over $\mathbf{Circ}$ and monoidal signature $(\mathbb N, {+}, 0)$, and equations:
	\\
	\textbf{Fork:} $\fork\circ v =v\otimes v.$
	\\
	\textbf{Join}: $\join\circ (v\otimes v') = v\sqcup v'.$
	\\
	\textbf{Stub}: $\g w \circ v= \mathit{id}_0,$ 
	\\
	\textbf{Gate}: $k\circ \bigotimes_{i=1,m}v_i=v$, such that whenever $v_i\sqsupseteq v_i'$ then $k\circ \bigotimes_{i=1,m}v_i\sqsupseteq k\circ \bigotimes_{i=1,m}v_i'$.	
\end{mydef}
We will call morphisms in this category \textit{combinational circuits.}

The first three model the fact that a fork duplicates a value, a join coalesces two values, and a stub discards anything it receives. The gate equations must cover all possible inputs to a gate~$k$ and their particular format entails that the output from a gate is always one of the original levels in~$\mathbf V$. Since $\mathbf V$ is a lattice, the monotonicity requirement is also expressible equationally.

It is known that, in a formal sense, the equality of morphisms in a free SMC corresponds to graph isomorphisms in the diagrammatic language~\cite{joyal1991geometry}, where diagrams are created by the operations of sequential composition ($\circ$), parallel composition ($\otimes$) and symmetry ($\g x_{m,n}$, the swapping of two buses with $m$ and~$n$ wires, respectively), governed by coherence equations. We will usually write composition in diagrammatical order $f\cdot g= g\circ f $. We write the identity (bus of width~$m$) $\mathrm{id}_m:m\rightarrow m$ as simply $m$. For simplicity we also write $\bigotimes_{i=1,m}f=f^m$, $\bigotimes_{i=1,m}f_i=\mathbf{f}$ and $\bigotimes_{i=1,m}v_i=\mathbf{v}$. For lack of space we will not enumerate the coherence equations here, since they are standard.
%
%
%

The \textit{Gate} axioms state that the behaviour of basic components is fully defined by their inputs, i.e. they are \textit{extensionally complete}.
By simple inductive arguments on the structure of morphisms we can establish that all circuits are in fact extensionally complete, i.e. for any circuit (not just gates) $f:m\rightarrow n$, for any values $v_i, 1\leq i\leq m$, there exist unique values $v'_j, 1\leq j\leq n$ such that $f\circ \bigotimes_{i=1,m}v_i=\bigotimes_{i=1,n}v_i'$. Intuitively this means that we only model \textit{local interactions}, abstracting away from global effects such as electromagnetic interference or quantum tunelling etc. 

We can further say that two circuits with the same input-output behaviour are \emph{extensionally equivalent}, and a simple inductive argument shows that this is a \emph{congruence}, i.e. it is an equivalence preserved by sequential and parallel composition. Therefore it makes sense to \emph{quotient} our category $\mathbf{CCirc}$ and create a new category $\mathbf{ECCirc}$ in which equivalent circuits are made equal. 

$\mathbf{ECCirc}$ has interesting additional categorical properties which aid reasoning. Two are of particular importance. The first one is that $\mathbf{ECCirc}$ is \emph{Cartesian}. The \emph{diagonals} are defined by $
\Delta_0 = 0$ and 
$\Delta_{n+1} =(\Delta_n \otimes \fork)\cdot(n\otimes \g x_{(1,n)}\otimes 1)
$, \textit{forks} of width~$n$. The diagonal has two important \emph{coherences} represented by the following diagram equalities valid for any diagram $f:n\rightarrow m$.
\begin{center}
	\includegraphics[scale=1]{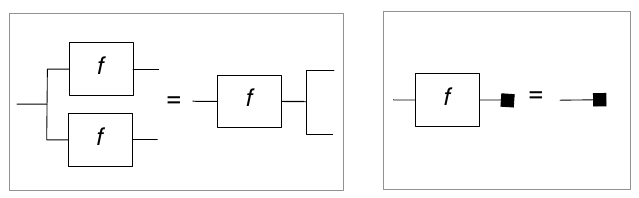}
	\\
	$\langle f,f\rangle=\Delta_n\cdot(f\otimes f)=f\cdot\Delta_m\qquad f\cdot \g w^m=\g w^m$. 
\end{center}
Another useful property is that $(\fork, \join, \g w, \bot)$ forms what is known as a \emph{bialgebra}, i.e. an algebraic structure in which $(\join, \bot)$ is a commutative monoid, $(\fork, \g w)$ is a co-commutative co-monoid, such that
$\join \cdot \fork = \fork^2\cdot(1\otimes \g x_{1,1}\otimes 1)\cdot \join^2$:
\begin{center}
	\includegraphics[scale=1]{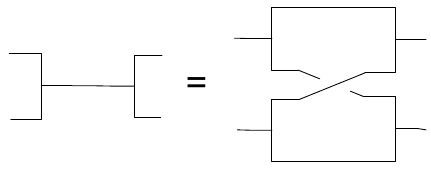}\\
\end{center}
Fork is a section and join a retraction,
$\fork\cdot \join=1$:
\begin{center}
	\includegraphics[]{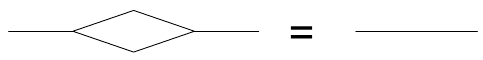}\\
\end{center} 
They are not generally isomorphisms, except for a special context:
\begin{myprop}\label{prop:pseudo}
	For any $f:m\rightarrow n+1$, 		$\Delta_m\cdot f^2\cdot(n\otimes (\join\cdot \fork)\otimes n)
	=\Delta_m\cdot f^2\cdot(n\otimes 1)^2$: 
	\begin{center}
		\includegraphics[]{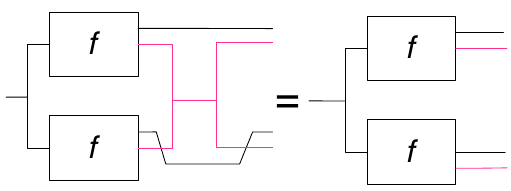}
	\end{center}
\end{myprop}
Another useful connector is the \textit{join} of width $n$, defined as $\nabla_0=0$ and $\nabla_{n+1}=(n\otimes\g x_{1,n}\otimes 1)\cdot(\nabla_n\otimes \join)$.

\subsection{Circuits with discrete delays}

\begin{mydef}
	\label{def:ccircd}
	Let $\mathbf{CCirc}_\delta$ be the category obtained by freely extending $\mathbf{ECCirc}$ with a new morphism $\delta:1\rightarrow 1$ subject to the following equations:
	\begin{compactitem}
		\item[\textbf{Timelessness}:] For any gate or structural morphism $k:m\rightarrow n$, $\delta^m\cdot k=k\cdot \delta^n$.
		\item[\textbf{Streaming}:] For any gate $k:m\rightarrow 1$ and levels $\mathbf v$, $( \delta^m\otimes\mathbf{v})\cdot\nabla_m\cdot k=(( \delta^m\cdot k)\otimes(\mathbf{v}\cdot k))\cdot \g \nabla_1$. 
		\item[\textbf{Disconnect}:] $\bot\cdot\delta=\bot.$
		\item[\textbf{Unobservable delay}:] $\delta\cdot\g w=\g w$.
	\end{compactitem}
\end{mydef}
\emph{Timelessness} means that compared to $\delta$, all other basic gates and structural morphisms compute instantaneously.  An immediate consequence is that delays can be propagated through combinational circuits, akin to \textit{retiming}~\cite{leiserson1991retiming}. \emph{Disconnect} means that the initial conditions of circuits is $\bot$, so that a wire that also promises to dangle later might as well be considered dangling already. The last rule expresses the same for dangling output wires. Diagrammatically, this is:
\begin{center}
	\includegraphics[scale=1]{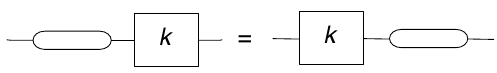}
\end{center}


The \emph{Streaming} axiom is more interesting, and it was one of the essential new axioms proposed in~\cite{GhicaJ16}. It is key to capturing the intuition of $\delta$ as a \emph{delay} operator. Mathematically, first observe that there are infinitely many morphisms of type $0\rightarrow1$ in $\mathbf{CCirc}_\delta$, not just the finitely many values. This is because expressions such as $v\cdot\delta$ do not reduce to a value. However, it can be shown that any expression built from values, $\delta$, and the structural morphisms can be transformed into a \emph{canonical form} which may be viewed as a \textit{sequence of values presented over time}, something that is called a \emph{waveform} in hardware design lingo. We write a waveform consisting of $n+1$ values as a list $s_n=v_n::v_{n-1}::\cdots::v_0$ where $v_n$ is the value that is currently visible, $v_{n-1}$ becomes visible in the next step, and so on. Formally, $s_0=v_0$ and $s_{n+1}=(s_n\cdot\delta\otimes v_{n+1})\cdot\join$. For example, the expression $v\cdot\delta$ corresponds to the waveform $\bot::v$; a value~$v$ is equal to (any of) the waveforms $v::\bot::\cdots::\bot$ which means that it is only available \emph{now} but no longer in the next time-step. As before, we write $\bigotimes_{i=i,m}s=s^m$ and $\bigotimes_{i=i,m}s_i=\mathbf{s}$. In the case of $m=2$, the equation is represented by the diagram:
\begin{center}
	\includegraphics[scale=1]{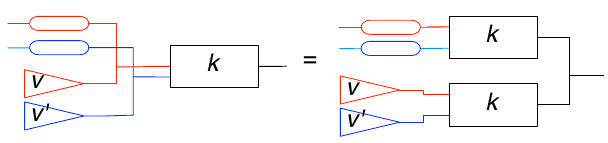}
\end{center}

The \emph{Streaming} axiom now tells us how a gate processes a waveform: we create two separate instances of the gate, one to process the immediate inputs and another to process the subsequent inputs. Applying it repeatedly to a given circuit allows us to determine the waveform that is produced at the output wires. We obtain:

\begin{mythm}[Extensionality]\label{thm:ext}
	Given any morphism $f$ in $\mathbf{CCirc}_\delta$, for any input waveform $\mathbf{s}$ there exists a unique output waveform $\mathbf{s'}$ such that $\mathbf s\cdot f=\mathbf {s'}$.
\end{mythm}
The proof is given elsewhere~\cite{GhicaJ16}, noting that Prop.~\ref{prop:pseudo} plays a key role. 

As in the case of circuits without delays, we can show that extensionality is a congruence and we can quotient by it, creating an \emph{extensional} category of circuits with delays, $\mathbf{ECCirc}_\delta$. It is then a routine exercise to show $\mathbf{ECCirc}_\delta$ is Cartesian, with the diagonal and terminal object defined the same way as in $\mathbf{ECCirc}$.


To conclude the section we will prove a generalisation of the Streaming axiom which will aid the formulation of the operational semantics. 
First note is a general diagrammatic reasoning principle which holds in all free symmetric monoidal categories. 

\begin{mylem}[Staging]
	\label{lem:pipe}
	Given a free SMC over a signature $\Gamma$, any morphism $f$ can be written as a sequence of compositions 
	$f = f_0\cdot f_1\cdots f_n $
	where $f_i$ is a tensor including exactly one non-identity morphism, $f_i = m\otimes k\otimes n$. 
\end{mylem}
\begin{proof}
	By diagrammatic reasoning, repeatedly rearranging the diagram as:
	
	\centering
	\includegraphics[scale=1]{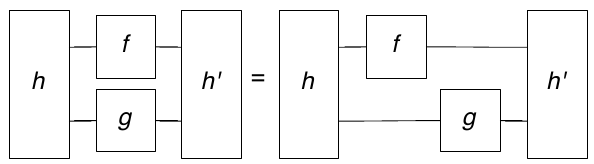}\\
\end{proof}

Let us call \textit{passive} a circuit which has no occurrences of a value. 

\begin{mylem}[Generalised Streaming]
	\label{lem:genstream}
	For any passive \textit{combinational} circuit $f:m\rightarrow n$, 
	$({\delta}^m\otimes m)\cdot\nabla_m\cdot f=( f\cdot \delta^n\otimes f)\cdot \g \nabla_n. $
\end{mylem}
\begin{proof}
	We make an extensional argument and use the Staging Lemma (\ref{lem:pipe}) and  induction on the number of stages. Each stage is combinational. For each stage, the property is immediate from Streaming and Timelessness. 
\end{proof}


\subsection{Circuits with feedback}	


\begin{mydef}
	Let $\mathbf{CCirc}_\delta^*$ be the category obtained from $\mathbf{ECCirc}_\delta$ by freely adding a trace operator. 
\end{mydef}
Diagrammatically, the trace operator applied to a diagram $f:m+k\rightarrow n+k$ corresponds to a feedback loop of width $k$, written $\tracei k{f}:m\rightarrow n$. Symmetric traced monoidal categories (STMC) satisfy a number of equations (coherences) which we will not enumerate for lack of space~\cite{joyal1996traced}. As before, their interpretation coincides with equality of diagrams (with feedbacks) up to graph isomorphism. 

As before, we are committed to an extensional view of circuits where the only observable is the input-output behaviour. In combinational circuits, with or without delays, the only way we can create a circuit with 0 outputs is by explicitly composing a circuit $f:m\rightarrow n$ with $\g w^n$. However, 0-output circuits can arise in more complicated ways in the presence of feedback, whenever all the outputs are fed back. For example, the diagram on the left can be reduced to just three unobserved inputs:
\begin{center}
	\includegraphics[scale=1]{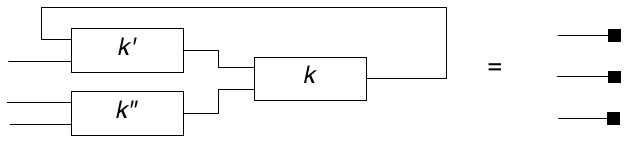}
\end{center}
Such equalities cannot be proved out of local interactions, so we will simply impose the equivalence of all 0-output circuits, an equivalence which is trivially a congruence. The new quotient category is called $\mathbf{OCirc}_\delta^*$. In this category all diagrams of shape $f:m\rightarrow 0$ are therefore equal which, categorically speaking, makes 0 a ``\emph{terminal object}''.

In general, in programs feedback corresponds to recursion and iteration, and syntactic models (operational semantics) of such programs involve creating two copies of the code recursed over. For example, the operational semantics of the Y-combinator as applied to some $G$ is 
$
YG = G(YG)
$.
A similar rule does not exist in general for SMTCs unless the category is also Cartesian. Such categories, also called \emph{data-flow categories}~\cite{cuazuanescu1990towards}, admit an \emph{iterator} defined for any $f:m+n\rightarrow n$:
\begin{center}
	\includegraphics[]{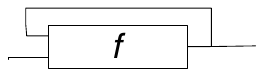}\\
	$\iter^n (f)=\tracei n{f\cdot(\Delta_n\otimes n)}:m\rightarrow n$
\end{center}
which satisfies the following equations:\\[1ex]
\emph{Naturality}:
$\iter((g\otimes n)\cdot f)=g\cdot\iter(f)$ for any $g:k\rightarrow m$.
\hfill \includegraphics[]{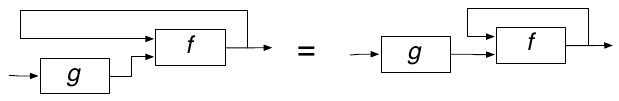}

\noindent\emph{Iteration}:
$\iter(f)=\langle m,\iter(f)\rangle\cdot f$ \hfill \includegraphics[]{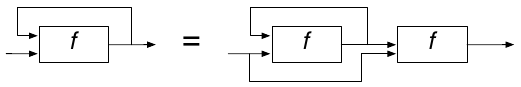} 

\noindent\emph{Diagonal}:
$\iter^n(\iter^n(f))= \iter^n((\langle n,n\rangle\otimes m)\cdot f)$.  \hfill \includegraphics[]{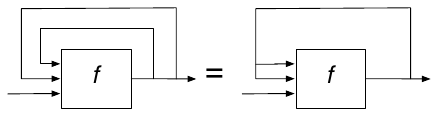}

Naturality justifies treating diagrams containing the iterator graphically, just as we are used to from SMTCs. Let us also state formally that the precondition for the iterator is fulfilled by the SMTC of circuits with feedback:

\begin{mythm}[\cite{GhicaJ16}]\label{thm:ocirccart}
	The category $\mathbf{OCirc}_\delta^*$ is Cartesian with diagonal~$\Delta_n$.
\end{mythm}


\newcommand{\V}{\mathbf{V}}
\newcommand{\Vcat}{\mathcal{V}}
\renewcommand{\S}{\mathbf{S}}
\newcommand{\Scat}{\mathcal{S}}
\newcommand{\sem}[1]{\llbracket{#1}\rrbracket}
\newcommand{\semV}[1]{\llbracket{#1}\rrbracket^{\Vcat}}
\newcommand{\semS}[1]{\llbracket{#1}\rrbracket^{\Scat}}
\newcommand{\cv}{\underline{v}}
\newcommand{\cbot}{\underline{\bot}}
\newcommand{\fix}{\mathrm{fix}}

\subsection{A model for $\mathbf{OCirc}_\delta^*$}\label{app:snd}
To conclude the section, we sketch the construction of a model for $\mathbf{OCirc}_\delta^*$ which will confirm the intuition we have alluded to above. It needs to be Cartesian and support the delay operator and iteration. The usual example of a traced SMC, sets and relations, is not Cartesian so a slightly more complex construction is required. We start with a basic model for combinational circuits based on the lattice $\V$ of values (Def.~\ref{def:circ}).

\begin{mydef}
	Let $\Vcat$ be the category whose objects are finite powers of~$\V$ and whose morphisms are monotone maps.
\end{mydef}

\begin{myprop}
	$\Vcat$ is cartesian, i.e., it has finite products.
\end{myprop}

For each basic gate~$k$, $\Vcat$ contains the concrete monotone function describing its input-output behaviour. The special morphisms \emph{join}, \emph{fork}, and \emph{stub} are represented by $\sqcup$, the diagonal into the product, and the unique map into the one-element lattice $\V^0$, respectively; all of these are monotone. Since $\mathbf{CCirc}$ is freely generated by the basic gates, we have:

\begin{myprop}
	There is a unique monoidal functor $\semV\cdot$ from $\mathbf{CCirc}$ to $\Vcat$ mapping the object $1$ of the former to $\V$ of the latter.
\end{myprop}

It is a trivial exercise to check that the equations given in subsection~2.1 are satisfied in~$\Vcat$, so $\semV\cdot$ factors through the extensional category $\mathbf{ECCirc}$.

To model the delay operators, we replace $\V$ by $\S=\V^\N$, the $\N$-indexed product of $\V$ with itself.
Its elements are \emph{streams} of elements of~$\V$. This is an infinite lattice, the order being defined componentwise, and it makes sense therefore to stipulate that morphisms preserve not just the order but also suprema of directed sets. Thus we have the category~$\Scat$ of finite powers of~$\S$ with Scott-continuous maps between them. As before, we immediately obtain that $\Scat$ is cartesian. The interpretation of gates and special morphisms is \emph{component-wise}:
$\semS{k}(s_1,\ldots,s_m)[n]=\semV{k}(s_1[n],\ldots,s_m[n])$, where we are reusing the semantics of gates inside~$\Vcat$. Values ($=$~levels) are interpreted as those streams that have the corresponding semantic value in first position and bottoms after that, which we write as~$\cv$. The following is obvious because the gate operations are defined componentwise:
\begin{myprop}
	The semantic gate functions $\semS{k}$ are Scott-continuous.
\end{myprop}

{\textit{Delay.}}
In the stream model we interpret delay by \emph{shift}:
\begin{displaymath}
\semS{\delta}(s)[n]=\left\{\begin{array}{ll}\bot&\text{if }n=0\\s[n-1]&\text{otherwise}\end{array}\right.
\end{displaymath}
Again, it is easily seen that this is a Scott-continuous operation from $\S$ to~$\S$. 

The soundness of the \emph{timelessness} rule (Def.~\ref{def:ccircd}) is now easily checked; it holds because gates act componentwise on streams and whether we shift the input streams or the output streams is immaterial. \emph{Streaming} is more interesting. In the case of $m=1$, we want to show
\begin{displaymath}
\semS{\delta\mathord{\otimes}v\cdot\join\cdot k}=\semS{(\delta\cdot k)\mathord{\otimes}(v\cdot k)\cdot\join}
\end{displaymath}
so let $v\in\V$ be a value and $s\in\S$ a stream. Consider first the head position of the output stream. The left hand side evaluates as
\begin{displaymath}
\begin{array}{rcl}
\semS{\delta\mathord{\otimes}v\cdot\join\cdot k}(s)[0] 
&=& \semS{\join\cdot k}(\semS\delta(s),\cv)[0] \\
&=& \semS{k}(\semS\delta(s)\sqcup\cv)[0] \\
&=& \semV k(\bot\sqcup v)\\
&=& \semV k(v)
\end{array}
\end{displaymath}
and the right hand side as
\begin{displaymath}
\begin{array}{rcl}
\semS{(\delta\cdot k)\mathord{\otimes}(v\cdot k)\cdot \join}(s)[0] 
&=& \sqcup(\semV k(\bot),\semV k(v))\\
&=& \semV k(v)
\end{array}
\end{displaymath}
where the last equality holds because by monotonicity we know that $\semV k(\bot)\sqsubseteq\semV k(v)$ and hence the supremum is the larger of the two. The same argument is used for all other positions $n>0$:
\begin{displaymath}
\begin{array}{rcl}
\semS{\delta\mathord{\otimes}v\cdot\join\cdot k}(s)[n] 
&=& \semV k(s[n-1]\sqcup\bot)\\
&=& \semV k(s[n-1])\\
&=& \semV k(s[n-1])\sqcup\semV k(\bot)\\
&=& \semS{(\delta\cdot k)\mathord{\otimes}(v\cdot k)\cdot\join}(s)[n]
\end{array}
\end{displaymath}

It is clear from the semantics of the delay operator that a circuit with $n$ delay operators takes account of (at most) the inputs ${s[m-n]},\ldots,{s[m-1]},{s[m]}$ when computing the $m$-th output value. Thus we have the following full characterisation of extensional equality:

\begin{myprop}\label{prop:exp}
	Two circuits $f$ and~$g$, which both contain no more than $n$ delay operators, are observationally equivalent if and only if they produce the same outputs for all waveforms of length up to~$n+1$.
\end{myprop}


{\textit{Feedback}.}
Iteration can be added to $\Vcat$ already. Let $f\colon\V^{m+n}\to\V^n$ be a monotone map. Define $\fix(f)\colon\V^m\to\V^n$ by
$
\fix(f)(\bar a)=\sSup_{k\in\N}\Phi(k,f,\bar a)
$,
where $\Phi(0,f,\bar a)=\bar\bot$ and $\Phi(k+1,f,\bar a)=f(\bar a,\Phi(k,f,\bar a))$ (and $\bar a$ is short for
$a_1,\ldots,a_m$). The definition of $\fix(f)$ in $\Scat$ is very similar, the only difference being that we start the iteration with the constant bottom stream $\cbot$ on each of the $n$~inputs.

In interpreting circuits with feedback, we let $\fix$ be the semantics of~$\iter$. 
\begin{myprop}
	$\fix$ satisfies the three equations for the iterator.
\end{myprop}
\noindent\emph{Remark.} This fact is mentioned (without proof) in Example~7.1.2 of \cite{hasegawa2012models} but for the convenience of the reader we include the argument.
\begin{proof}
        The first equation, \emph{Naturality}, is a straightforward calculation:
        \begin{displaymath}
          \fix((g\otimes n)\cdot f)(\bar a)=\sSup_{k\in\N}\Phi(k,(g\otimes n)\cdot f,\bar a)
        \end{displaymath}
        \begin{displaymath}
          (g\cdot\fix(f))(\bar a)=\fix(f)(g(\bar a))=\sSup_{k\in\N}\Phi(k,f,g(\bar a))
        \end{displaymath}
        The two chains on the right are term-wise identical as one sees by induction over~$k$: For $k=0$ we get $\bar\bot$ in both cases and for $k+1$ we have
        \begin{displaymath}
          \Phi(k+1,(g\otimes n)\cdot f,\bar a)=((g\otimes n)\cdot f)(\bar a,\Phi(k,(g\otimes n)\cdot f,\bar a))=f(g(\bar a),\Phi(k,(g\otimes n)\cdot f,\bar a))\quad\mbox{and}
        \end{displaymath}
        \begin{displaymath}
          \Phi(k+1,f,g(\bar a))=f(g(\bar a),\Phi(k,f,g(\bar a)))
        \end{displaymath}
        which are equal by induction hypothesis.

	The second equation, \emph{Iteration}, holds because all maps in $\Scat$ are Scott-continuous and therefore $\sSup_{k\in\N}g^k(\bar\bot)$ is a fixpoint of~$g$:
	\begin{displaymath}
	f(\bar a,\fix(f)(\bar a))= g(\sSup_{k\in\N}g^k(\bar\bot))=\sSup_{k\in\N}g^{k+1}(\bar\bot)=\sSup_{k\in\N}g^k(\bar\bot)
	\end{displaymath}
        The left-hand side of the third equation, \emph{Diagonal}, can be re-written into a system of simultaneous equations by \emph{Beki\v c's rule}, and these precisely describe the right-hand side.
\end{proof}
Note that in the case of $\Vcat$ the computation of the supremum takes place in the finite lattice $\V^n$. This means that the terms $\Phi(k,f,\bar a)$ can only take finitely many different values, or in other words, that the supremum is obtained after finitely many iterations. This agrees with intuition: As the circuit is powered up, the potential in all wires stabilises after a (very short) finite delay; there is no possibility of infinite oscillation. In the case of $\Scat$ the same holds \emph{for every component} of the output, although the (mathematical) iteration will typically take infinitely many steps to produce all components of the output stream.

A Cartesian category admits an iterator (satisfying the three equations listed above) if and only if it is traced. This result appears in \cite[Section~7.1]{hasegawa2012models} where it is credited to Martin Hyland. The following is now an immediate consequence of this and the fact that $\mathbf{CCirc}_\delta^*$ is the \emph{free} symmetric traced monoidal category over
$\mathbf{Circ}_\delta$:

\begin{myprop}
	There is a unique traced monoidal functor $\semS\cdot$ from $\mathbf{CCirc}_\delta^*$ to $\Scat$ mapping the object $1$ of the former to $\S$ of the latter.
\end{myprop}

In the presence of delay and feedback, even an input waveform of finite length can produce an infinite output stream. However, such an output stream must be eventually periodic. This is because the output at any point in time is completely determined by the contents of the delay components and the inputs at that time. A circuit can only contain finitely many delay components (say $n$~many) and so the ``internal memory'' can only assume $k=\left|\V\right|^n$-many different states. Any input stream that is longer than~$k$ must lead to at least one internal state re-occurring. In other words, if we provide such a circuit with all possible input streams of length~$k$ (there are $\left|\V\right|^k$-many) then we are guaranteed that the circuit has assumed all possible internal states it can ever assume. This means that if we test for the output from  one more input after all these initial waveforms then we have tested the circuit under all possible internal configurations. We have shown:
\begin{myprop}\label{prop:exp-exp}
	Two circuits with feedback $f$ and~$g$, which both contain no more than $n$ delay operators, are observationally equivalent if and only if they produce the same outputs for all waveforms of length up to~$\left|\V\right|^n+1$.
\end{myprop}
Propositions \ref{prop:exp} and~\ref{prop:exp-exp} give an exponential and a superexponential upper bound, respectively, for the number of tests required for checking the equivalence of circuits. They provide the background against which we present the \emph{much more efficient} diagrammatic operational semantics below.

The semantics also suggests other equations, not considered in this paper, for example, the following always holds
\begin{displaymath}
  \sSup_{k\in\N}g^k(\bar\bot)=\sSup_{k\in\N}g^{k+1}(\bar\bot)=\sSup_{k\in\N}g^k(g(\bar\bot))\;.
\end{displaymath}
On diagrams, this equations does not seem to lead to productive re-writes, but it does point to additional reasoning principles.

\section{Diagrammatic operational semantics}\label{sec:diags}

The results of the previous section establish a powerful framework for algebraic reasoning about circuits. However, this framework is not equally useful for \textit{automatic} reasoning and cannot implement a reasonable operational semantics. 

The first obstacle is the functoriality property of the tensor, which lacks directionality. Consider the circuit corresponding to the boolean expression $t\wedge f\wedge t$, where the constants involved satisfy the obvious equations. 
This diagram can be specified in several ways. Some of the specifications, e.g. $(((t\otimes f)\cdot{\wedge})\otimes t)\cdot {\wedge}$ have the immediately identifiable redex $(t\otimes f)\cdot{\wedge}=f$ which reduces the overall expression to $(f\otimes f)\cdot {\wedge}$, which reduces to $f$. However, the same circuit can be equivalently written as $(t\otimes f\otimes t)\cdot ({\wedge}\otimes id)\cdot{\wedge}$ which has no obvious redex.
\begin{center}
	\includegraphics[scale=1]{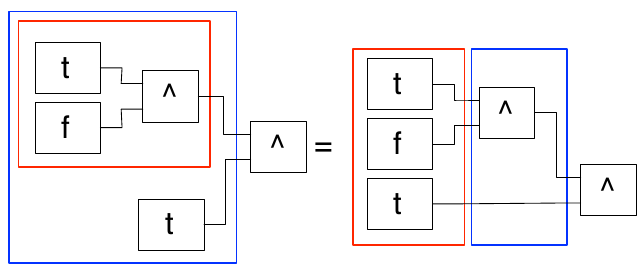}
\end{center}
Finding redexes in such structural diagram specifications is computationally prohibitive and an unsuitable operational semantics. The alternative is to exploit the connection between monoidal categories in general, and traced monoidal categories in particular, and certain graphs. This idea has been analysed in depth recently~\cite{DBLP:conf/lics/BonchiGKSZ16}.

We will give a concrete presentation of the graphs following Kissinger's \textit{framed point graphs}, which are a free (strict) symmetric traced monoidal category~\cite[Thm.~5.5.10]{Kissinger}. To make the presentation more accessible we will elide some of the categorical technicalities in \textit{loc. cit.} and give a more direct presentation.

Let a labelled directed acyclic graph (LDAG) be a DAG $(V,E)$ equipped with a partial labelling function $f:V\rightharpoonup L$. Let a labelled interfaced DAG (\textit{LIDAG}) be a labelled DAG with two distinguished lists of unlabelled nodes representing the ``input" and ``output" interfaces. \textit{LIDAGs} can be visualised as
\begin{center}
\includegraphics[scale=1]{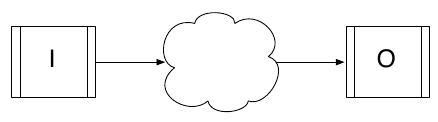}
\end{center}
Unlabelled nodes are called \textit{wire nodes} and edges connecting them are called \textit{wires}. A \textit{wire homeomorphism}~\cite[Sec.~5.2.1]{Kissinger} is any insertion or removal of wire nodes along wires, which does not change the shape of the graph. Two \textit{LIDAG}s are considered to be equivalent if they are graph isomorphic up to renaming vertices and wire homeomorphisms. The quotienting of \textit{LIDAG}s by this equivalence gives us \textit{framed point graphs} (FPG)~\cite[Def.~5.3.1]{Kissinger}. The algebraic specifications of the diagrams associated with the expression $t\wedge f\wedge t$ mentioned above all correspond to the (same) framed point graph with empty input interface and 1-point output interface.
\begin{center}
\includegraphics[scale=1]{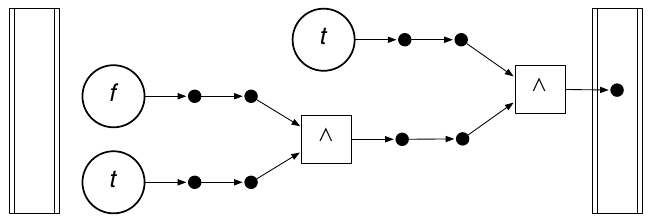}
\end{center}
This representation solves the problem raised by the functoriality of the tensor. Moreover, the graph representation can be made canonical by eliminating all redundant wire nodes efficiently (in linear time in terms of the size of the graph). In the PROP with FPGs as morphisms, composition, tensor and trace can be represented visually as below. 

Sequential composition of two FPGs where the size of the output of the first matches the size of the input of the second is defined by identifying the output list and the input list of the two graphs. Since FPGs are equal up to renaming of vertices, the names of the wires can be chosen so that the composition is well defined. The unlabelled input and output nodes become wire nodes in the composition.
\begin{center}
	\includegraphics[scale=1]{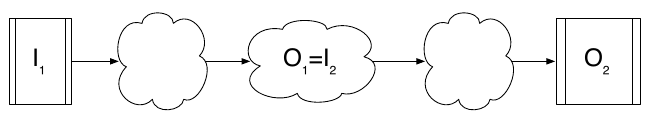}
\end{center}
The tensor is the disjoint union of the two graphs. It is always well defined since graphs are identified up to vertex renaming.
\begin{center}
	\includegraphics[scale=1]{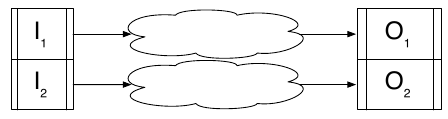}
\end{center}
The trace operator picks the head nodes of the input and output lists of points, makes them wire nodes, and connects them. 
\begin{center}
	\includegraphics[scale=1]{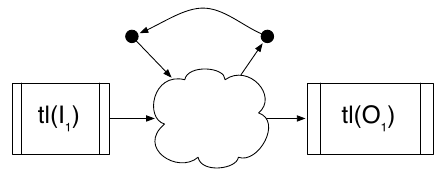}
\end{center}
The graph representation provides a solution for dealing with the functoriality of the tensor, but the presence of feedback raises a new, additional problem. Suppose that we deal with a graph which includes several iterations, e.g. $\iter (f)\cdot \iter (g)$.
\begin{center}
	\includegraphics[scale=1]{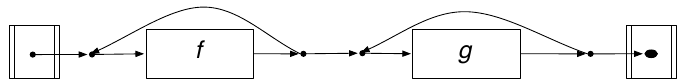}
\end{center}
This graph raises two computationally difficult questions. The first one is how we identify feedback patterns efficiently so that we can apply the iteration axiom. The second one is, if there are several instances of the iteration unfolding axiom that can be applied, what is the schedule of applying them? Without a good (linear time) solution to the first problem we cannot claim that we have a genuine operational semantics. Without a good solution to the second problem we run into technical problems of confluence. Diagrammatic representation alone is no longer the solution. 

The main contribution of this section is showing how to solve these two problems.
\begin{mylem}[Global trace form]\label{lem:gbltrc}
	For any morphism $f$ in a free STMC there exists a trace-free morphism $\hat f$ such that $\text{Tr}^n(\hat f)=f$ for some $n\in\mathbb N$.
\end{mylem}
\begin{proof}
	The proof is by diagrammatic reasoning\footnote{This can be done algebraically, of course, but it is much less clear.}. The same diagram can be created either by applying a trace locally or globally for any $g:m+m'\rightarrow p_1+p_2+p_3$, $f':n+p_2\rightarrow n+p_2'$, $h:p_1+p_2'+p_3\rightarrow m+n'$:
	\begin{multline*} 
	\text{Tr}^m(g;(p_1\otimes\text{Tr}^n(f')\otimes p_3);h) =
	\text{Tr}^{m+n}((n\otimes g);(\mathsf x_{n,p_1}\otimes p_2\otimes p_3);(p_1\otimes f'\otimes p_3);(\mathsf x_{p_1,n}\otimes p_2'\otimes p_3);(n\otimes h)).
	\end{multline*}
	\begin{center}
		\includegraphics[scale=1]{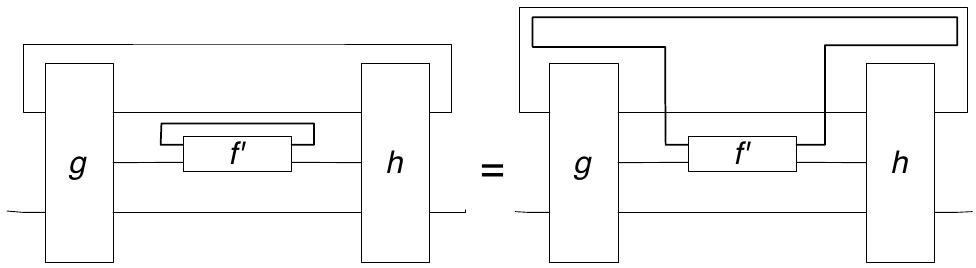}
	\end{center}	
	Since diagrams are finite, this rewrite can be repeated to write the diagram as a global trace applied to a  trace-free morphism~$\hat f$.  
\end{proof}
So a diagram with feedback loops can always be rewritten as single, global, feedback loop. In the graph we can maintain a distinguished subset of known \textit{feedback} wire nodes so that the feedback loops can be immediately identified. This can be done compositionally just by keeping track of the feedback wire nodes in sequential composition, tensor and trace. 

The most interesting case is of the trace:
\begin{center}
	\includegraphics[scale=1]{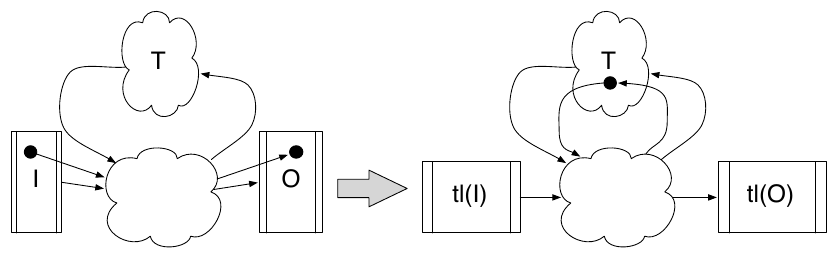}
\end{center}
The wire homeomorphism rules allow us to place exactly one wire node in the set of feedback wire nodes. 
By maintaining the feedback wire nodes explicitly we can ensure two useful invariants. First, the rest of the graph is a DAG. Second, for each feedback wire node there is precisely one incoming and one outgoing edge. We call these graphs \textit{trace-framed point graphs} (TFPGs). Note that feedback wire nodes must not be entirely removed as wire homeomorphisms are applied. Feedback edges that bypass the set of feedback wire nodes are legal, but break the TFPG form. Maintaining these restrictions is computationally trivial (constant overhead).


We are now in a position to define the diagrammatic operational semantics as a graph-rewriting system in which each rule can be applied efficiently, in linear time as a function of the size of the graph. 

Given a categorical signature $\mathbf{Circ}$ we use the levels and the gates as labels, along with labels denoting discrete delays and stubs as well as a set of labels $\g i_k, k\in\mathbb N$ denoting distinguished input ports for gates. 
For readability we usually omit the labels and for join, fork, and its input ports.
\begin{center}
	\includegraphics[scale=1]{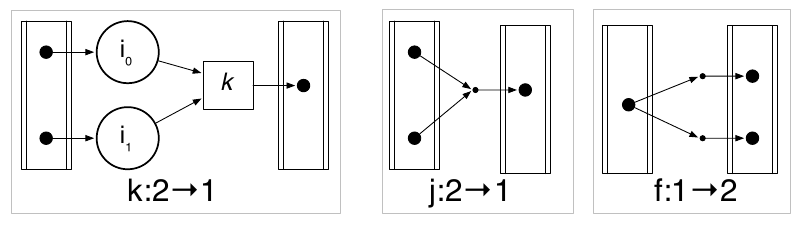}
\end{center}
A circuit specified as a morphism is a TFPG. We add the following rewrite rules corresponding to the categorical axioms. Note that each rule can be applied in linear time, including the identification of the redex. 

\subsection{Combinational rules}
The diagrammatic semantics will be given as a collection of graph-rewriting rules. We give the rules in an informal diagrammatic style, but a formalisation in an established formalism such as DPO~\cite{corradini1997algebraic} is a standard exercise.

\textit{Constant.} The basic rule is the reduction of a gate with known input levels. If $\mathbf{v}\cdot k = v'$ then 
\begin{center}
	\includegraphics[scale=1]{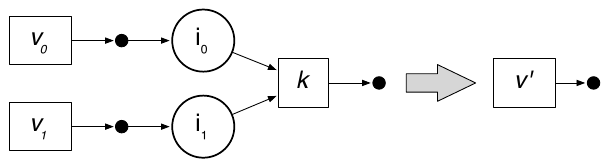}
\end{center}

\textit{Enhanced Constant Rules.} Besides the basic equations for constants, more equations can be proved by extensionality in which reductions can be carried out without all input values being present. For example, $\mathit{true}\vee x=\mathit{true}$ or $\mathit{true}\wedge x= x$. These equations are admissible in the rewrite system.

\textit{Fork.} In contrast, the forking of a wire means copying the value attached to it. Forking has $\g w$ as a co-unit.  
\begin{center}
	\includegraphics[scale=1]{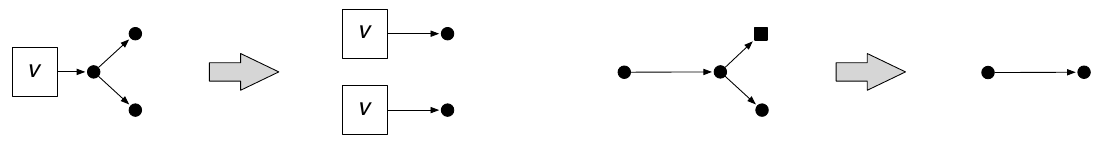} 
\end{center}

The rules below are for \textit{waveforms}, where $v::s$ is the circuit diagram 
	\includegraphics[scale=1]{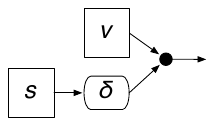}.
Note that waveforms  are simple sub-graphs that can be identified in the overall circuit diagram in linear time. 

\textit{Streaming.} For any values $v,v'$, waveforms $s,s'$ and constant $k$
\begin{center}
	\includegraphics[scale=1]{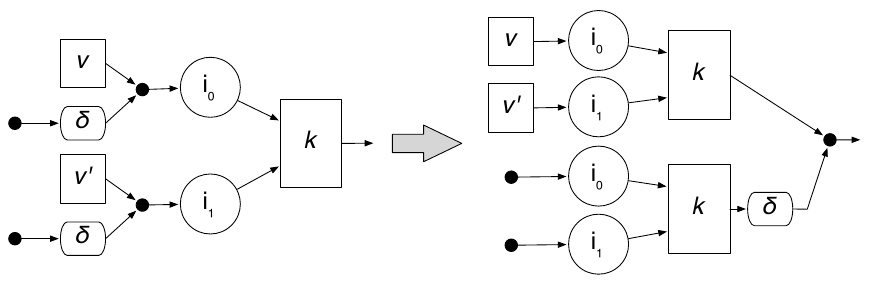}
\end{center}


\textit{Stub.} We omit the label and the input port of the stub $\mathsf w$, for readability. For any constant $k:m\rightarrow 1$,
\begin{center}
	\includegraphics[scale=1]{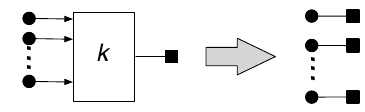}
\end{center}
We call the rewrite rules above the \textit{local rewrite rules}. A TFPG where no local rewrite rules apply is in \textit{canonical form.} 
\begin{myprop}
	The local rewrite rules are sound. 
\end{myprop}
If the graph on the left is a representation of a morphism $f$ then the graph on the right is a representation of a morphism $g$ and $f=g$ in $\mathbf{ECCirc}$. The proof is immediate. 
\begin{mylem}\label{lem:term}
	The local rewrite rules are terminating.
\end{mylem}
\begin{proof}
	TFPGs are finite. All rules reduce the sum of path lengths to the input. In the case of input we assume that the new redex corresponding to the term $(v\otimes v')\cdot k$ is reduced as well. 
\end{proof}
\begin{mylem}\label{lem:snd}
	The local rewrite rules are strongly confluent.
\end{mylem}
\begin{proof}	
	Given a TFPG so that more than one local rewrite rule is applicable we first note that for all rules other than \textit{Stub} the redex subgraphs must be disjoint, so they commute trivially. But the \textit{Stub} rule also commutes, obviously, with any other rule. It follows that the system is strongly confluent. 
\end{proof}
\begin{mylem}[Progress]\label{lem:prog}
	A circuit $f:0\rightarrow n, n\neq 0$ without traces or delays is either a {value} or the TPFG associated with it has redexes. 
\end{mylem}
\begin{proof}
	A trace-free circuits of type $0\rightarrow n$ must have form $\mathbf v \cdot g$ for some $g$. If the TPFG representation of $g$ consists just of identities then by wire homeomorphisms we have a representation of a value. If the TPFG contains a constant, or a node structure corresponding to join or fork, the input nodes must connect directly to constants so there are redexes. 
\end{proof}
From Lem.~\ref{lem:snd} and~\ref{lem:prog} it follows that
\begin{mythm}\label{thm:redux}
	Given a circuit $f:0\rightarrow m, m\neq 0$ in $\mathbf{ECCirc}$ the local rules will always rewrite its TPFG representation in a finite number of steps into a TPFG representation of a \textbf{value}~$\mathbf v$ such that $f=\mathbf v$. 
\end{mythm}

\subsection{Feedback and delay}
We now need to add rules to ensure to deal with delays which occur in arbitrary places in the circuit, not just in waveforms. For example, a circuit such as $(t\otimes f)\cdot (1\otimes \delta)\cdot \wedge$, in TFPG representation, does not have any redex because of how the delay is placed. Dealing with the delays requires a complex rule which takes into account the presence of the trace. The trace and the delay must be dealt with together because of the following result which allows us to write any circuit in what we will call \textit{global-delay form.}
Note Lem.~\ref{lem:genstream} does not hold for combinational circuits with values. However, the following holds:
\begin{mylem}\label{lem:pass}
	For any combinational circuit $f:m\rightarrow n$ there exists a passive circuit $\tilde f$ such that $f=(m\otimes\mathbf v)\cdot\tilde f$ for some~$\mathbf{v}$.
\end{mylem}
\begin{proof}
	The proof is immediate by diagrammatic reasoning, by applying this transformation for each occurrence of a value $v$:
	\begin{center}
		\includegraphics[scale=1]{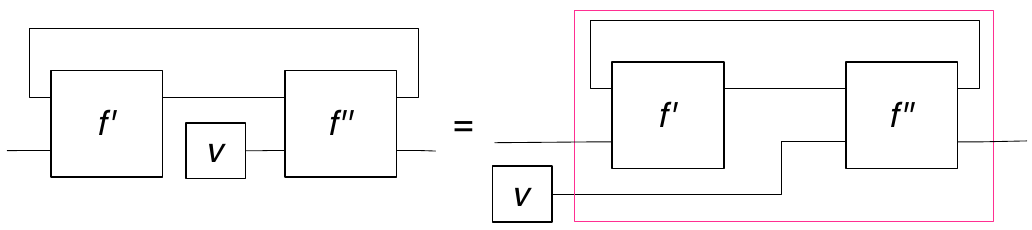}
	\end{center}
\end{proof}
We call the application of the transformation in this lemma the \textit{passification} of the circuit. 
\begin{mylem}\label{lem:gdf}
	Any circuit $f$ in $\mathbf{OCirc}_\delta^*$ can be written as $f=\tracei{m}{(\delta^n\otimes p)\cdot f'}$ for some trace-free, delay-free circuit $f'$, $m,n,p\in\mathbb{N}$. 
\end{mylem}
\begin{proof}
	The proof is diagrammatic. All diagrams can be rewritten in the following, isomorphic, shape:
\begin{center}
	\includegraphics[scale=1]{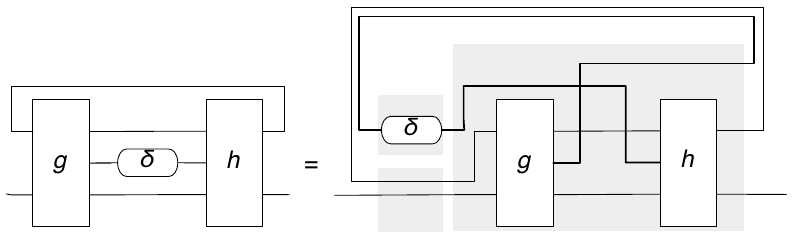}
\end{center}	
\end{proof}
\textit{Trace-Delay.} The most complex rule is the unfolding of the global trace, which also handles the delays. We will explain it before stating it. 

The first step is to derive an unfolding axiom for trace from the unfolding axiom for iteration, by expressing trace in terms of iteration. This is possible using the co-monoid $\fork$ and its co-unit~$\g w$~\cite{hasegawa2012models}:
\begin{center}
	\includegraphics[]{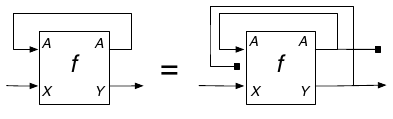}
\end{center}
\[\tracei{A}{f}=\iter^{A\otimes Y} ((id_A \otimes \g w_Y \otimes id_X) \cdot f) \cdot (\g w_A \otimes id_Y).\]
The right-hand side is an iteration which we unfold, and then simplify: 
\begin{center}
	\includegraphics[]{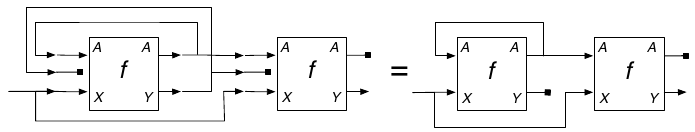}
\end{center}
If used as an operational semantics, the rewrite rules will apply to representations of circuits which have known inputs, so their input interface is empty (closed circuits). We will concentrate on this usage of the rewrite system, for now. In the case of optimisation for partial evaluation the more general unfolding provided above can be used. 

In the case of closed circuits the unfolding of the trace is shown in Fig~\ref{fig:tracedelay}. 
\begin{figure}
	\centering
	\includegraphics[]{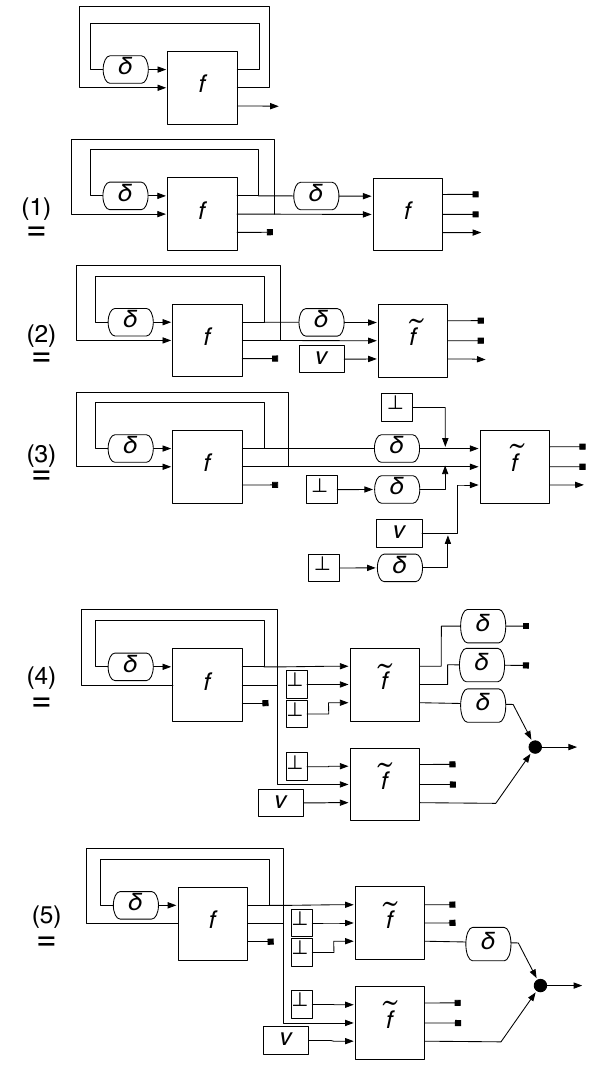}
	\caption{Unfolding a circuit in global-delay form}
	\label{fig:tracedelay}
\end{figure}
Step (1) represents the unfolding of the trace. The delays are all managed separately, so step (2) passifies the combinational circuit $f$ (Lem.~\ref{lem:pass}). Step (3) uses $\bot$ as the unit of the join-monoid along with the \textit{Unobservable Delay} axiom, to bring the circuit to a form where \textit{Generalised Streaming} (Lem.~\ref{lem:genstream}) can be applied (Step (4)). A final simplification removes redundant delays which are not observable (Step (5)). A final step (not shown) moves the delays on the output of $f$ to restore the global-delay form, using Lem.~\ref{lem:gdf}. The resulting circuit can be represented as a TFPG. 

%
%

\textit{Wire homeomorphisms.} We know that wire homeomorphisms can be applied efficiently, bringing the graph to a normal form~\cite[Lem.~5.2.10]{Kissinger}. However, we have prohibited the elimination of feedback nodes using these wire homeomorphisms. Instead we use an explicit rule that ``unwinds'' feedback loops if and only if they reach a value $v$. This rule results in the removal of a wire feedback node. A similar rule propagates stubs across feedback loops. 
\begin{center}
	\includegraphics[scale=1]{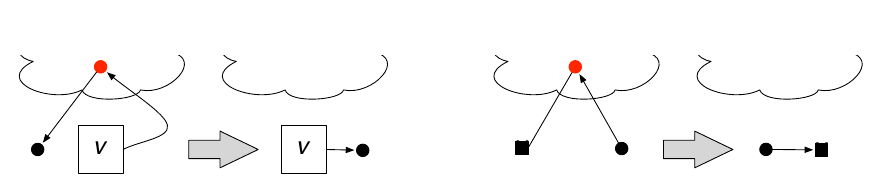}
\end{center}
The wire homeomorphism is a ``tidy-up'' rule which does not play a key role in the proof of our main results about the rewrite system, but is practically useful as it reduces the size of the graph, a form of {garbage collection}. 

We define the overall rewriting system as a cycle of local rewrites until canonical form is reached, followed by trace-delay unfoldings. This system is obviously not terminating, which is consistent with the fact that circuits with feedback can generate infinite waveforms. E.g.,
$
\iter(v::1)=v::\iter(v::1)=v::v::\iter(v::1)=\cdots.
$
which diagrammatically corresponds to (up to graph isomorphism and wire homeomorphism, to keep the representation simple):
\begin{center}
	\includegraphics[scale=1]{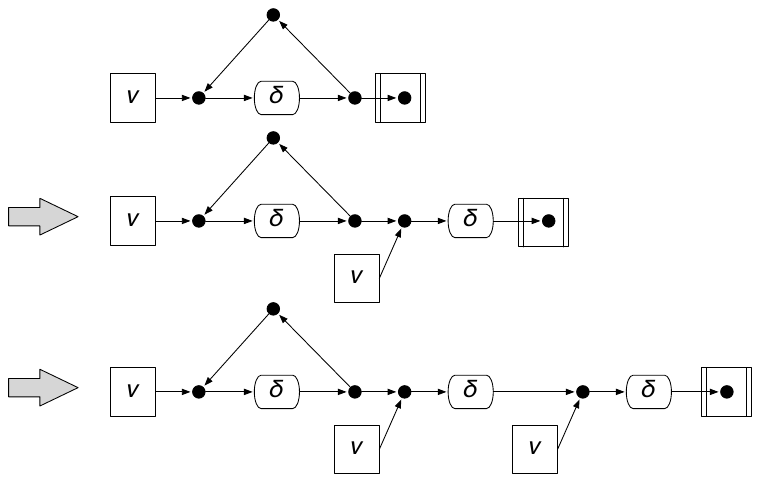}
\end{center}
\subsection{Productivity}
In a circuit of the form $v::f=(v\otimes(f\cdot \delta))\cdot\join$ value $v$ will be observed \textit{before} whatever the behaviour of $f$ is, since $v$ is \textit{instantaneous} whereas $f$ is guarded by a delay. We call such circuits \textit{productive}, and we add a \textit{labelled} rewrite rule to simplify productive circuit by removing the produced value
\[ v::f \stackrel v\Longrightarrow f.\]
This rule is sound because the sub-circuit $v::-$ can never be part of any redex.
So the example above can be written as
\begin{align*}
\iter(v::1)=v::\iter(v::1)\stackrel{v}{\Longrightarrow}\iter(v::1).
\end{align*}

However, we note circuits need not be productive in general. There exist circuits where unfoldings never reduce to shape $v::f$. Take, for example, the unfolding of $t\cdot \iter(\wedge)$:
\begin{center}
	\includegraphics[scale=1]{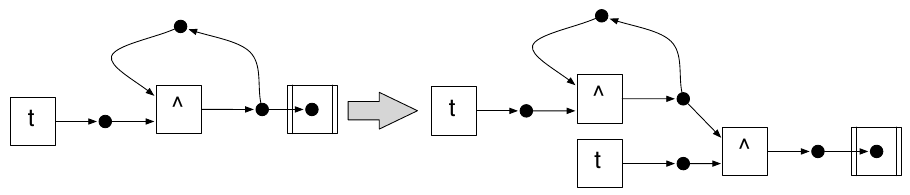}
\end{center}
This is a well known problem caused by a genuine \textit{instant  feedback} loop between the output and one of the inputs of the gate. If a circuit has no instant feedback loops, it is guaranteed to be productive. 
\begin{mydef}
	We say that a circuit has \textit{delay-guarded feedback} if its global-delay form is $\tracei{m}{\delta^m\cdot f}$. 
\end{mydef}
If a circuit has delay-guarded feedback loops then it is productive. In fact it implements a Mealy automaton. 
\begin{mythm}\label{thm:dgc}
	Delay-guarded circuits with no inputs are productive. Given the TPFG representation of a delay-guarded feedback, the rewrite system will produce a TPFG graph representing a circuit $v::g$ in a finite number of steps.  
\end{mythm}
\begin{proof}
	The proof can be expressed diagrammatically as the sequence of equal diagrams in Fig.~\ref{fig:dgc}.
\begin{figure}
	\centering
	\includegraphics[]{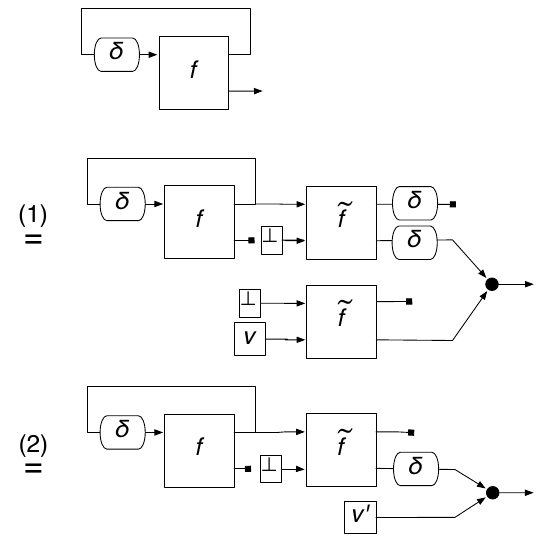}
	\caption{Proof of Thm.~\ref{thm:dgc}.}
	\label{fig:dgc}
\end{figure}
Step 1 is the unfolding of the global-delay form trace, with no instant feedbacks. Step 2 is the extensionality property of $(\bot^m\otimes \mathbf v)\cdot f$ in $\mathbf{ECirc}$, which must reduce to a value $\mathbf v'$, combided with the non-observability of delay axiom. Because the final diagram represents a circuit of the form  $\mathbf v'::-$ it is by definition productive.

The rewrite system will successfully produce the final circuit because in computing the canonical form, the sub-circuit $\bot^m\cdot f$ is delay free and trace free so Lem.~\ref{lem:prog} ensures that it will reduce to value $\mathbf v$. 
\end{proof}
With a delay, the unproductive example $t\cdot \iter{(\wedge)}$ becomes the productive $t\cdot\iter{((\delta\otimes 1)\cdot\wedge)}$ which, after a series of routine calculations can be shown to reduce to a productive circuit of shape~$\bot :: f$.

Note that the delay-guarded feedback condition is sufficient but not necessary. An interesting example of circuits with non-delay guarded feedback which are productive are the cyclic combinational circuits which we discuss in Sec.~\ref{sec:ccc} below.

To wrap up the operational semantics we also give a necessary and sufficient non-productivity criterion, to prevent needless unfoldings of the circuit. 

\begin{mythm}
	If a circuit has the shape in Fig.~\ref{fig:tracedelay} and is unproductive then all further unfoldings of the trace will be unproductive. 
\end{mythm}
\begin{proof}[Proof outline]
Let a circuit as in Fig.~\ref{fig:tracedelay} (Step 5) be unproductive. Suppose that we unfold the trace again. The resulting circuit will look as in the diagram below. The sub-circuit drawn as a cloud is irrelevant for productivity because all its outputs are guarded by delays -- so we can omit it. 
\begin{center}
\includegraphics[]{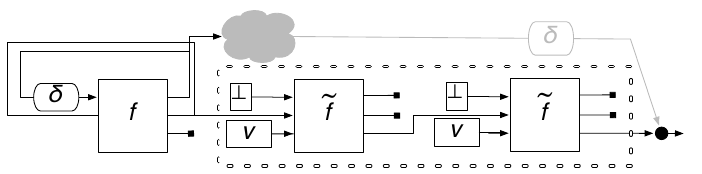} 
\end{center}

If this circuit is to be productive then the circuit framed by the dotted line must reduce to a value. But this contradicts the hypothesis, because if the first (on the left) occurrence of the sub-circuit $\tilde f$ can reduce to values, then it could have reduced to values before the unfolding of the circuit. So the unfolded circuit must be also unproductive.
\end{proof}
\subsection{Example: Cyclic combinational circuits}

A challenging class of circuits, which cannot be handled by standard tools, are combinational circuits with feedback which is not delay guarded~\cite{DBLP:conf/iccad/Malik93}. Consider Boolean circuits with \textit{and} and \textit{or} gates. The following is an example of such a circuit: 
\begin{center}
	\includegraphics[]{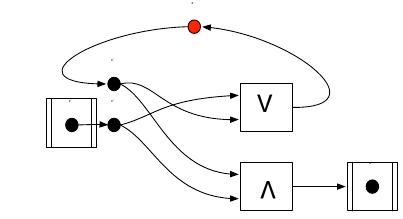}
\end{center}
Closing the circuit by applying a boolean value at the input makes it possible to apply the diagrammatic operational semantics, using the enhanced equational rewrite rules: 
\begin{center}
	\includegraphics[]{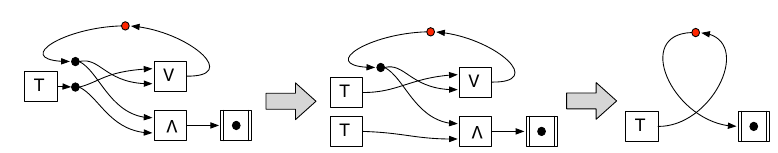}
\end{center}
Finally, we unfold the (superfluous) loop:
\begin{center}
	\includegraphics[]{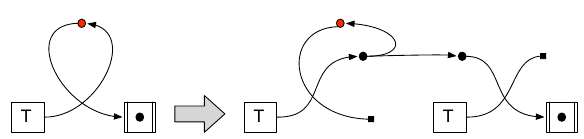}
\end{center}
The circuit above reduces to \textit{True}, by applying the co-unit of fork and the \textit{Stub} rule several times.

\section{Specialising abstract digital circuits}\label{sec:ex}
If we are not using the rewrite rules as an operational semantics, and so are not concerned with productivity issues, we can apply the reduction rules to open and to parametrised circuits. This gives us a basis for powerful partial evaluation-like optimisations of circuits. This is a new contribution with potentially interesting practical applications.

\subsection{Abstract cyclic combinational circuits}\label{sec:ccc}
Consider the circuit represented by the TFPG below (highlighting feedback nodes in red), where the gate $m$ is a \textit{multiplexer} and $F, G$ are abstract circuits. For readability we omit the input labels of the multiplexer.
\begin{center}
	\includegraphics[scale=1]{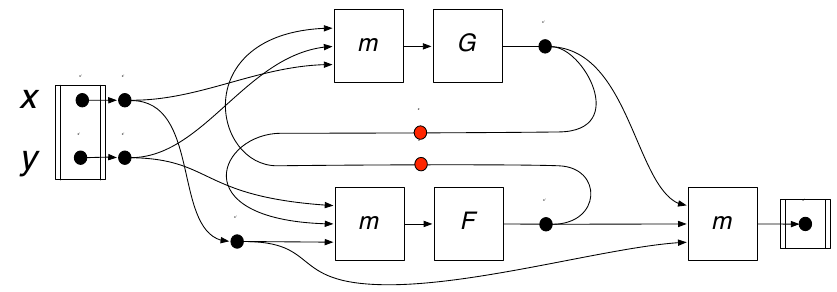}
\end{center}
This circuit, presented in~\cite{DBLP:conf/iccad/Malik93}, implements the operation \textit{if x then F(G(y)) else G(F(y))}. The circuit has no delays so the feedback loops are combinational, so they cannot be handled by conventional circuit analysis tools. However, the multiplexers are set up so that no matter what the value applied at $x$, the residual circuit is feedback-free. The false feedback loops in the circuit are only a clever way to reuse the two abstract circuits $F$ and $G$.

Consider the case when $x$ becomes $\g t$:
\begin{center}
	\includegraphics[scale=1]{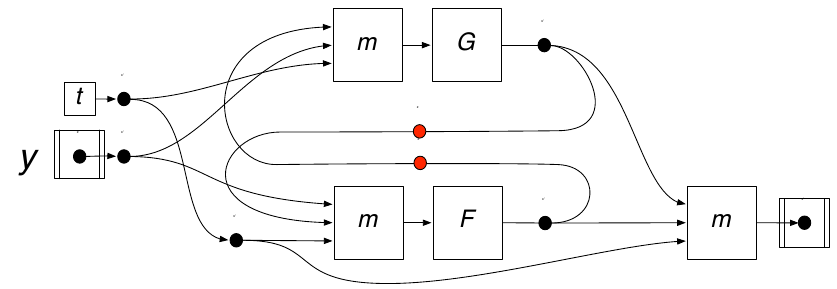}
\end{center}
Fork and  enhanced equational rewrite rules for $m$ lead to:
\begin{center}
	\includegraphics[scale=1]{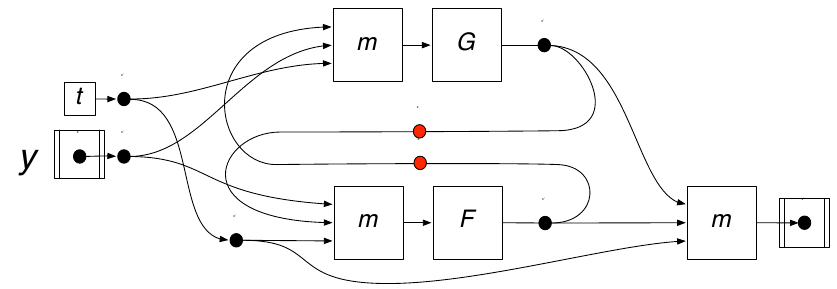}
\end{center}
The \textit{Stub} rule repeatedly applied results in a circuit which can be written as, to emphasise the residual feedback loop:
\begin{center}
	\includegraphics[scale=1]{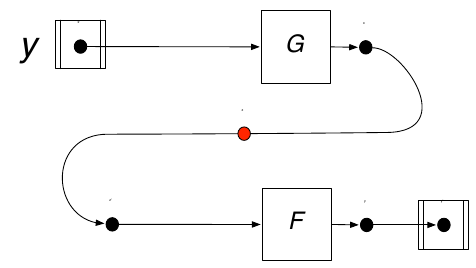}
\end{center}
To the naked eye the circuit above is obviously feedback-free, being equal to $G\cdot F$. However, the rewrite rules have no way to eliminate feedback wire nodes in this TPFG and we need to unfold the circuit:
\begin{center}
	\includegraphics[scale=1]{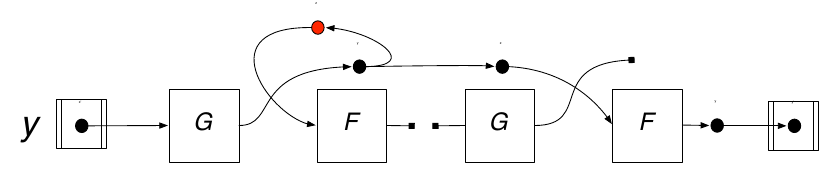}
\end{center}
The \textit{Stub} rule will remove the first occurrence of $F$ and the second occurrence of $G$, resulting, as expected, in $G\cdot F$.


\subsection{Pre-logical circuits}
The diagrammatic semantics can also model operationally transistor-level circuits, which is also a new development. 

The circuit framework is general enough to allow operational reasoning about digital circuits at a level of abstraction below logical gates, for example metal-oxide-semiconductor field-effect transistor (MOSFET) circuits. In \textit{saturation mode} such transistors can be considered to take on a discrete set of values which, depending on the circuit and the analysis, can be four-valued (\textit{high impedance} $<$ \textit{high, low} $<$ \textit{unknown.}) or six-valued (\textit{high impedance} $<$ \textit{weak high, weak low}; \textit{weak high} $<$ \textit{strong high}; \textit{weak low} $<$ \textit{strong low}; \textit{strong high, strong low} $<$ \textit{unknown.}). Unlike Boolean logic, where the wire-join construct is not used, in a transistor circuit output wires are joined, and the semantics of the wire-join is that of the value-lattice join operator. 

We will work in the six-value lattice $\bot$ (high impedance), $\g h$ (weak high), $\g H$ (strong high), $\g l$ (weak low), $\g L$ (strong low), $\top$ (unknown). We will take the (idealised) nMOS and pMOS transistors as the basic gates. The nMOS transistor ($\g n:2\rightarrow 1$) works like a low-activated switch, but it only allows low current to flow. High current can flow, but is much diminished. The defining equations are
\begin{align*}
(\g L\otimes \g L)\cdot \g n & = \g L &
(\g L\otimes \g l)\cdot \g n & = \g l \\
(\g L\otimes \g H)\cdot \g n & = \g h &
(\g L\otimes \g h)\cdot \g n & = \bot \\
(\g H\otimes v)\cdot \g n & = \bot &
(v \otimes v')\cdot \g n & = \top, \text{ if } v\neq \g{H,L.} \\
(v\otimes \top )\cdot \g n & = \top
\end{align*}
The pMOS transistor ($\g p:2\rightarrow 1$) is activated by the $\g h$ value and allows low $\g L$ value to pass, so the equations are the converse of their nMOS counterparts.

When implementing a logical gate in MOSFET we want $\g H$ to correspond to \textit{true} and $\g L$ to false. The correct behaviour of a gate must keep this representation without, e.g. producing $\top$ or weak output $\g h, \g l$.

A very simple circuit is the \textit{inverter} (\textsf{inv}), given in conventional schematics and as a TFPG. The syntactic description is $\g{inv}=\fork\cdot (1\otimes \g h\otimes 1\otimes \g l)\cdot (\g p\otimes \g n)\cdot \join$.
\begin{center}
	\includegraphics[scale=0.25]{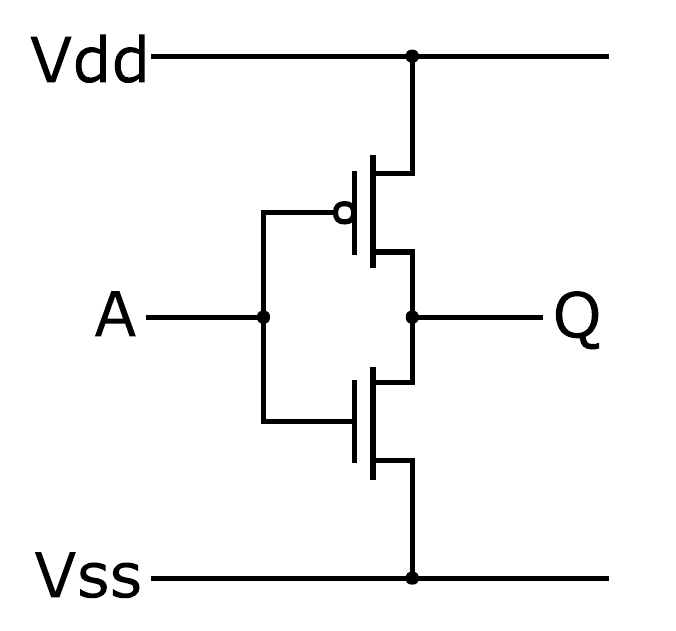}\qquad 
	\includegraphics[]{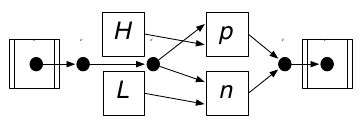}
\end{center}
We can see via a sequence of rewrites that this circuit correctly maps $\g H$ to $\g L$ and vice-versa. For example:
\begin{center}
	\includegraphics[scale=1]{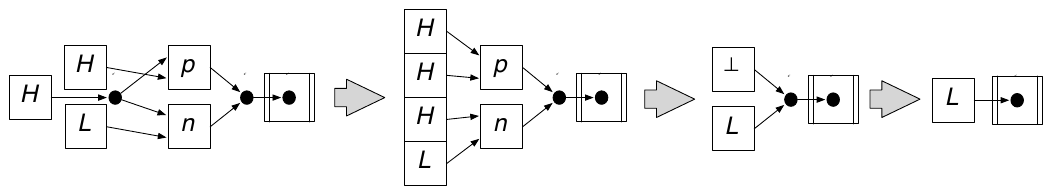}
\end{center}

%

Let us now revisit the example of the previous section, but with the multiplexer implemented down to transistors. Let \textsf{pass} be a pass-through gate and \textsf{m} a multiplexer:
\begin{align*}
\mathsf{pass} &= \fork^2\cdot(\g{inv}\otimes 3)\cdot (1\otimes \g x \otimes 1)\cdot (\g p\otimes \g n)\cdot \join\\
\mathsf{m} &= (\fork\otimes 2)\cdot (1\otimes \g x \otimes 1)\cdot (2\otimes \g{inv}\otimes 1)\cdot \g{pass}^2\cdot \join.
\end{align*}
The resulting TFPG has too many nodes to reduce by hand but we have implemented a prototype tool for partial evaluation by rewriting, available for download\footnote{\url{https://github.com/AliaumeL/circuit-syntax}}.

The abstract circuit of the previous section is represented as a TFPG in the first graph in Fig.~\ref{fig:maliktrans}. The residual circuit after partial evaluation is shown as the second graph. It is interesting to note that the MOSFET version of the circuit leads to a different residual circuit compared to the more high level circuit of the previous section. The reason is that we do not use any enhanced rewriting rules, so the residual circuit contains some irreducible pass-through gates. 
\begin{figure*}[t]
	\includegraphics[scale=0.11]{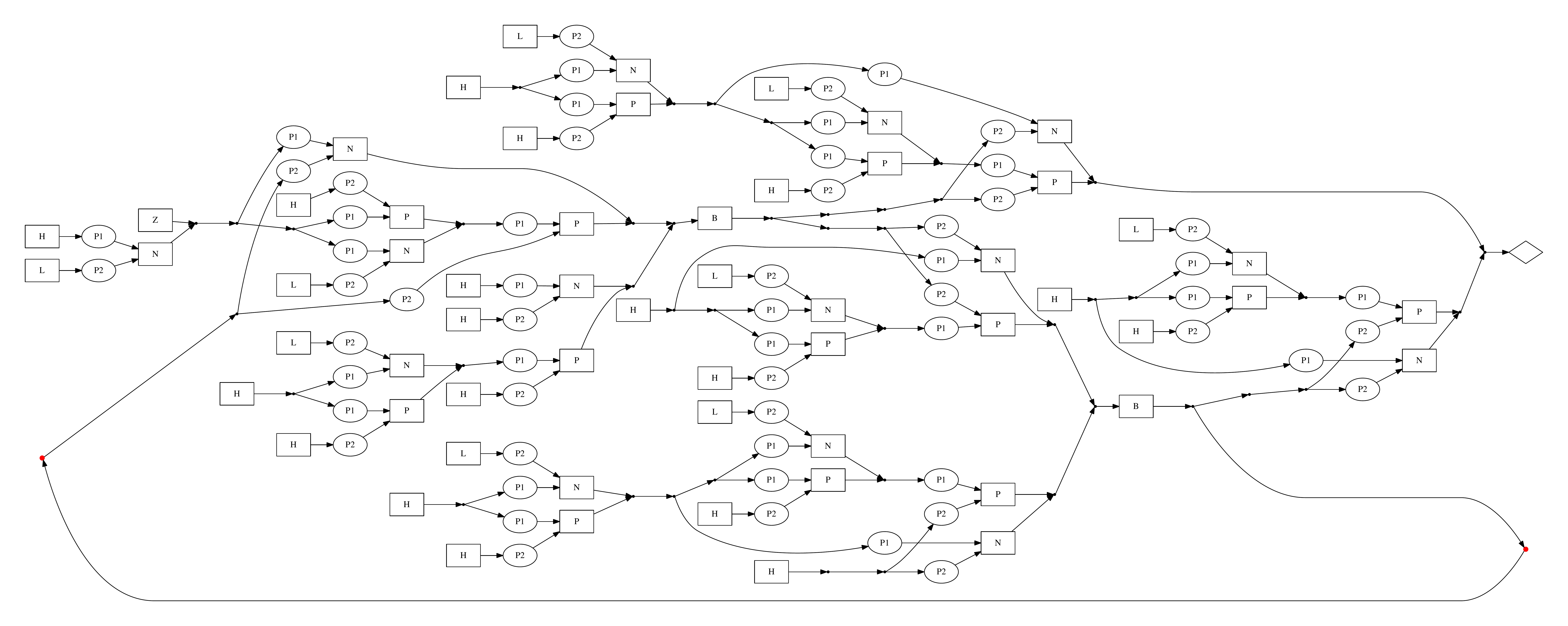} \raisebox{12ex}{$\Rightarrow$}
	\raisebox{8ex}{\includegraphics[scale=0.11]{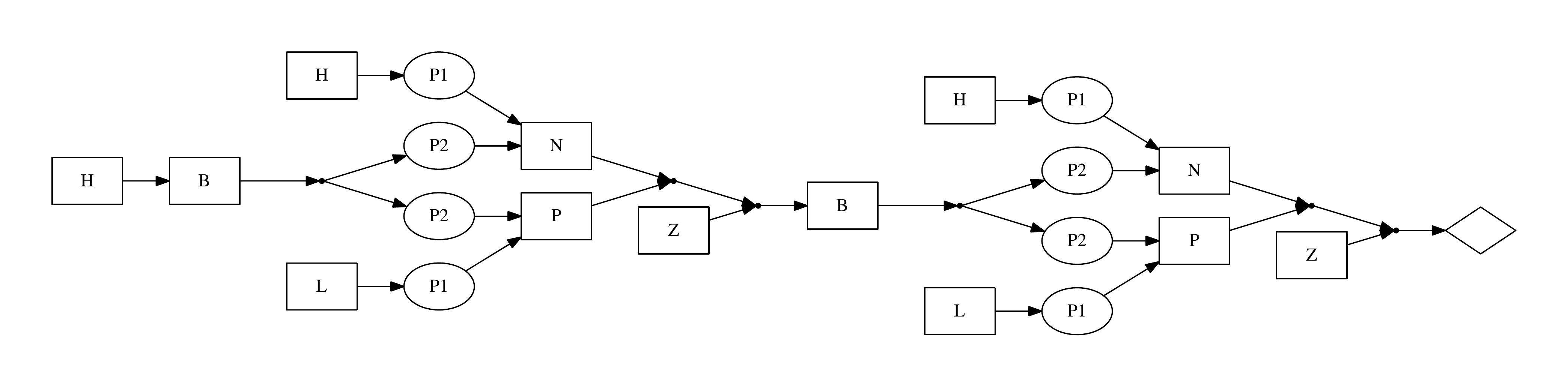}}
	\caption{MOSFET implementation of the circuit in Sec.~\ref{sec:ccc}, and the residual circuit after partial evaluation}
	\label{fig:maliktrans}
\end{figure*}

\section{Conclusion, related and further work }

Some theoretical ingredients we have used in this work have been around for quite a while and it is perhaps somewhat surprising that they have not been put together for a coherent operational and diagrammatic treatment of digital circuits. 
Our Thm.~\ref{thm:ocirccart} implies that $\mathbf{OCirc}_\delta^*$ is a Lawvere theory~\cite{DBLP:journals/iandc/EilenbergW67} with trace, also known as an iteration theory~\cite{elgot1975monadic}, a concept which has been studied extensively~\cite{bloom1993iteration}, leading to recent connections with rewrite systems~\cite{DBLP:conf/rta/Hamana16}. The relation between trace and iteration has also been studied before in a somewhat similar categorical setting~\cite{DBLP:journals/entcs/Hasegawa02}. The connection between Lawvere theories and PROPS has also been recently studied~\cite{DBLP:conf/cmcs/BonchiSZ16}.

We have been in particular inspired by the deep connections between monoidal categories and diagrams~\cite{selinger2010survey} which \emph{inter alia} have been used in the modelling of quantum protocols~\cite{DBLP:conf/lics/AbramskyC04} and signal-flow graphs~\cite{DBLP:conf/popl/BonchiSZ15}. Some contrasts are quite interesting. Unlike in quantum protocols, all digital circuits with no inputs and no outputs are equal whereas in quantum computing they correspond to \emph{scalars}, which allow quantitative aspects to be expressed. Should we have taken a similar direction we could have included quantitative aspects such as power consumption in our formalism, but we would have lost the diagonal property. Obviously, two copies of a circuit will at least sometimes consume more power than one copy.

The signal-flow graphs in \cite{DBLP:conf/popl/BonchiSZ15} are linear and reversible, which is not the case for digital circuits. Without elaborating the mathematics too much, a key difference between their model and ours can be illustrated by the following equality, involving the interaction between fork, join, and disconnected wires, as a trace can be created out of a fork and a join:
\begin{center}
	\includegraphics[scale=1]{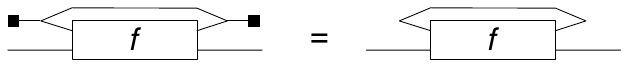}
\end{center}
Of course, by comparison, in our setting the directionality of the wires never changes, so the correct equality is:
\begin{center}
	\includegraphics[scale=1]{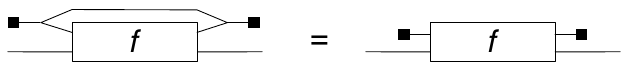}
\end{center}
These simple diagrammatic equations above truly capture the essential difference between \textit{electric} and \textit{electronic} circuits!

It is also quite surprising that despite major early progress in the algebraic treatment of circuits, \cite{DBLP:conf/lfp/Sheeran84,Luk1993}, this line of work has not come earlier to a systematic conclusion. But the contribution of our work is not merely assembling off-the-shelf components. The \textit{Streaming} axiom is new, and the fact that it generalises to arbitrary passive combinational circuits is a crucial ingredient for our work. To make the unfolding of iteration computationally tractable, the diagrammatic representation required a non-obvious canonical form, which must be easy to compute. Without it our earlier  semantics~\cite{GhicaJ16} cannot be used as an effective operational semantics.

Beyond the scholarly context and technical innovations, we are most excited about the potential applications of our work. Cyclic combinational circuits are a litmus test for circuit modelling theories and we hope the reader can appreciate that in our framework their model is elementary. For comparison, there are few theories that can handle such circuits, and they demand a significant level of mathematical sophistication~\cite{DBLP:journals/fmsd/MendlerSB12}. The true potential of our method is unleashing of symbolic, operational and syntactic methods, such as partial evaluation, for reasoning about and optimising circuits, methods which proved so effective in programming languages.

There is work to be done, from exploring theoretical questions to implementing  more efficient circuit-rewriting tools. The most interesting theoretical questions involve the extensionality and characterisation of circuits with feedback. We do not have yet equivalents of Thm.~\ref{thm:ext} for circuits with delays and feedback. These theorems would play an important role in dealing with circuits with instant feedback which are currently unproductive in the graph rewrite system, such as $t\cdot\iter{(\land)}$ or infinitely productive such as $\iter((1\otimes v)\cdot \delta)$, by developing meta-theoretic reasoning principles for such graphs. In the concrete category (Sec.~\ref{app:snd}) we can see that the first circuit is equal to $\bot$ while the second can be considered a ``normal form'' for infinite waveforms. Since all the operations in this setting are finite-state ($\mathbf{V}$ is finite) a general, complete and efficient framework should be achievable. This is currently in progress.

A more long term development could see these ideas applied to other computational structures which are dataflow categories, such as reactive programming~\cite{wan2000functional} or arrows~\cite{hughes2000generalising}, giving a diagrammatic operational semantics more abstract than that of the underlying programming language. 
\\[2ex]
\textit{Acknowledgements.} We thank George Constantinides and Alex Smith for feedback and suggestions.


\bibliographystyle{abbrv}

\end{document}